\documentclass[preprint,11pt]{elsarticle}

\journal{ }
\usepackage{graphicx}
\usepackage{pifont,latexsym,ifthen,amsthm,rotating,calc,textcase,booktabs}
\usepackage{amsfonts,amssymb,amsbsy,amsmath}
\numberwithin{equation}{section}
\usepackage[title]{appendix}
\usepackage[utf8]{inputenc}
\usepackage{url}
\usepackage{mathtools}
\usepackage{extarrows}
\usepackage{calc}
\usepackage[colorlinks,linkcolor=blue]{hyperref}
\newtheorem{theorem}{Theorem}[section]
\newtheorem{lemma}[theorem]{Lemma}

\newtheorem{assumption}[theorem]{Assumption}
\newtheorem{proposition}[theorem]{Proposition}

\newtheorem{problem}{RH problem}[section]
\usepackage{overpic}
\usepackage{extarrows}
\newtheorem{main theorem}{Theorem}
\usepackage{natbib}
\usepackage{todonotes}
\usepackage{float}
\usepackage{subfigure}
\numberwithin{figure}{section}
\usepackage{hyperref}
\biboptions{sort&compress}
\hypersetup{
	colorlinks=true, 
	linktoc=all,     
	linkcolor=blue,  
	citecolor=blue
}

\begin{document}
	
	
	
	\begin{frontmatter}
		\title{Soliton Shielding of the focusing modified KdV equation}
		

		\author[inst2]{Ruihong Ma}

		\author[inst2]{Engui Fan$^{*,}$  }
		
		\address[inst2]{ School of Mathematical Sciences and Key Laboratory of Mathematics for Nonlinear Science, Fudan University, Shanghai,200433, China\\
			* Corresponding author and e-mail address: faneg@fudan.edu.cn  }
		
		
		
		

		\begin{abstract}
We consider soliton gas solutions of the modified Korteweg-de Vries (mKdV) equation, where the point spectrum of the condensate is located within a bounded domain in the upper half-plane. We first demonstrate that when the domain is a quadrature and the soliton density is an analytic function, the corresponding deterministic soliton gas coincides with a finite number of solitons, which we call this effect soliton shielding. When the domain is an ellipse and the soliton density is analytic, the corresponding deterministic soliton gas reduces the spectral data to the segment joining the foci. The initial datum of this Cauchy problem is asymptotically step-like oscillatory, described by a periodic elliptic function as \( x \to +\infty \), and it vanishes exponentially fast as \( x \to -\infty \).
		\end{abstract}
		
	\end{frontmatter}
	\tableofcontents
	
	\section{Introduction }

This paper concerns with  the soliton shielding  for the   focusing modified Korteweg-de Vries (mKdV) equation
\begin{align}\label{eq:mkdc}
  &  q_t+6q^2q_x+q_{xxx}=0,\quad (x,t)\in\mathbb{R}\times\mathbb{R}^+,
\end{align}
with  nonstandard initial data $q(x,0)$  stemming   from the infinite soliton limit. The mKdV equation has emerged in various physical contexts, including the propagation of ultrashort few-optical cycle solitons in nonlinear media \cite{HF03}, anharmonic lattices \cite{HO92}, Alfven waves \cite{TH69}, ion acoustic solitons \cite{SW84}, traffic jams \cite{TT99}, Schottky barrier transmission lines \cite{VJ01}, thin ocean jets \cite{EL94}, among others. Mathematically,the mKdV hierarchy exihibit interesting integrable structures \cite{ZJ94,ZJ97,Monvel,Germain,ZJ02}. The mKdV equation
 (\ref{eq:mkdc}) admits a soliton solution of the form \cite{KM82}
\begin{equation}\label{eq:soliton}
q(x,t)=\pm2\kappa\,{\rm sech}\left[2\kappa\left(x-24\kappa^2t-x_0\right)\right],
\end{equation}
where the quantity $4\kappa^2$ is the wave velocity, $\kappa\in\mathbb{R}_+$ and $x_0$ is a phase shift.
The multi-soliton solutions of the   mKdV equation
 (\ref{eq:mkdc}) were  derived  by  Horita \cite{R72}, Wadati-Ohkuma \cite{KM82},  Schuur \cite{PC86} and Zhang-Yan \cite{ZY}.

Solitons are elementary, localized solutions to nonlinear evolution equations. They manifest as individual entities, embodying traveling wave solutions that maintain their shape during propagation. Alternatively, solitons can emerge as collectives, evolving as an ensemble and asymptotically decomposing into separate solutions and potential sub-ensembles of solitons. Following the discovery of soliton ensembles, their interpretation as particle-like entities has been a source of novel investigations.

 The notion of soliton gas, an infinite statistical ensemble of interacting solitons, was first introduced by Zakharov \cite{VEZ71}, who derived a kinetic equation for a ``rarefied" gas of Korteweg-de Vries (KdV) solitons by considering the modification of the soliton velocity due to the position shifts in its pairwise collisions with other solitons in the gas. The generalization of Zakharov's kinetic equation to the case of the KdV soliton gas of finite density was obtained  by considering the infinite-phase, thermodynamic-type limit of the Whithan modulation equation \cite{G.A. El}.  The kinetic theory for solitons has been extended to the cases of focusing Nonlinear Schr\"odinger (NLS) equation \cite{GE05}, defocusing and resonant NLS equation \cite{TG21}, and to the case of breather gasses in the NLS equation \cite{GE20} and mKdV equation \cite{MT23}. The finite-gap theory derivation in \cite{G.A. El} served as a motivation for a more intuitive, physical construction of the kinetic equation for a dense soliton gas of the focusing NLS equation in \cite{ElGAM}.  A RH problem describes a soliton gas  as the limit of a finite $N$-soliton configuration as $N\to+\infty$, and established rigorous asymptotics of the KdV potential in several different regimes  \cite{MT21}.
	
The notable emergent characteristic observed is that, as \( N \to \infty \), for specific types of domains and densities, a phenomenon known as ``soliton shielding'' occurs, wherein the gas behaves as if it were a finite number of solitons. In \cite{MT223}, the authors demonstrate that when the domain is a disk and the soliton density of NLS equation is an analytic function, the corresponding deterministic soliton gas unexpectedly yields a one-soliton solution, with the point spectrum located at the center of the disk. When the domain is an ellipse, soliton shielding reduces the spectral data to the soliton density concentrating between the foci of the ellipse. The $\bar{\partial}$-problem of the NLS equation on the domain, as discussed in \cite{KT08,MK19,BG83}, can be reduced to a classical Riemann-Hilbert (RH) problem with discontinuities along a collection of arcs. This reduction is also an effect of soliton shielding, as the effective soliton charge can be diminished from a two-dimensional domain to a collection of arcs.
		
From \cite{GJ20}, we realize that general solution to the focusing mKdV equations will consist of solitons moving to the right, breathers traveling in both directions. In \cite{MT23}, the authors focus exclusively  on the discrete spectrum associated with solitons $i\zeta_j$, where $\zeta_j > 0$ for $j = 1, \dots, N$. The soliton gas is generated from a continuum of poles that accumulate on the intervals $(-i\eta_2, -i\eta_1) \cup (i\eta_1, i\eta_2)$, with $0 < \eta_1 < \eta_2$, and with positive norming constants. They derive an expression for the solution of the soliton gas combined with a trial soliton in terms of a Fredholm determinant and analyze its asymptotic behavior for large times.


 The aim of this paper is to demonstrate the soliton shielding effect.  Our attention is directed towards cases where the spectrum is given by
		\[
		\mathcal{Z} = \left\{\{z_k, c_k\}_{k=1}^{N_1}, \{\kappa_j, h_j\}_{j=1}^{N_2}\right\},\,\,\, N=N_1+N_2,
		\]
		with the conditions that the soliton has \(\kappa_j = i\zeta_j\) where \(\zeta_j > 0\), and the breather has \(z_k = a_k + ib_k\) with \(a_k> 0\) and \(b_k > 0\). Furthermore, it is required that the spectrum converges on simply connected bounded domain \(\mathcal{D}\), which are located within the upper complex plane.
For domains \( \mathcal{D} \) that are both sufficiently smooth and simply connected, the boundary \( \partial\mathcal{D} \) is characterized by the Schwarz function \( S(z) \) associated with \( \mathcal{D} \), which fulfills the relationship
			\[
			\bar{z} = S(z).
			\]
			This Schwarz function can be analytically continued to a maximal domain \( \mathcal{D}^0 \subset \mathcal{D} \). The characteristics of \( \mathcal{D}^0 \) are detailed below:
			\begin{enumerate}[\(\bullet\)]
				\item In the case of quadrature domains, \( \mathcal{D}^0 \) is obtained by removing a finite set of points from \( \mathcal{D} \).
				\item For other domain classes, the set \( \mathcal{D} \setminus \mathcal{D}^0 \) is composed of a collection of smooth arcs.
			\end{enumerate}
			


The article is organized as follows: In Section \ref{sec2}, we first establish the deterministic soliton gas for the mKdV equation \ref{eq:mkdc}, and then we discuss the soliton shielding effect within the domain \(\mathcal{D}\), which may be a quadrature domain or an ellipse, as detailed in Sections \ref{sec222} and \ref{sec2222}. Finally, in Section \ref{secinitial}, we show that when the domain \(\mathcal{D}\) is an ellipse, the soliton shielding effect leads to an initial condition characterized by step-like oscillatory behavior.		
 \section{Soliton shielding }\label{sec2}
    \subsection{The deterministic soliton gas for the mKdV equation}

  In this section, our goal is to derive the soliton shielding for a domain \(\mathcal{D}\) that is either a quadrature or an ellipse. Initially, we consider the deterministic soliton gas for the mKdV equation as the number of solitons \(N\) approaches infinity. This derivation is conducted under the assumption that the norming constants \(c_j\) are scaled as \(\mathcal{O}(1/N)\). By employing the inverse scattering transform, which reconstructs the \(N\)-soliton solution from the spectral data \(\mathcal{Z}\), we can establish the RH problem for the pure \(N\)-soliton solution. The RH problem is characterized as follows \cite{GJ20}.
	\begin{problem}\label{RH21}
		Find an analytic function $M(z):\mathbb{C}\to SL_2(\mathbb{C})$ with the following properties.
		
		\begin{enumerate}
			\item $ M(z)=I+\mathcal{O}(z^{-1})$ as $z\to\infty$.
		
			\item $M(z)$ has simple poles at each soliton  $\kappa_j\in{\mathcal{{Z}}}$ at which
			\begin{align*}
				&\underset{\kappa=\kappa_j}{{\rm Res}} M(z)=\underset{\kappa\to\kappa_j}{\lim}M(z)\begin{pmatrix}
					0&0\\
					h_je^{2i\theta(z)}&0
				\end{pmatrix},\\
				&\underset{\kappa=\bar\kappa_j}{{\rm Res}} M(z)=\underset{\kappa\to\bar\kappa_j}{\lim}M(z)\begin{pmatrix}
					0&-\bar h_je^{-2i\theta(\bar z)}\\
					0&0
				\end{pmatrix}.
			\end{align*}
		$M(z)$ has simple poles at each breather $z_k\in{\mathcal{{Z}}}$ at which
				\begin{align*}
				&\underset{z=z_k}{{\rm Res}} M(z)=\underset{z\to z_k}{\lim}M(z)\begin{pmatrix}
					0&0\\
					c_je^{2i\theta(z)}&0
				\end{pmatrix},\\
				&\underset{z=\bar z_k}{{\rm Res}} M(z)=\underset{z\to \bar z_k}{\lim}M(z)\begin{pmatrix}
					0&-\bar c_je^{-2i\theta(\bar z)}\\
					0&0
				\end{pmatrix},\\
					&\underset{z=-\bar z_k}{{\rm Res}} M(z)=\underset{z\to -\bar z_k}{\lim}M(z)\begin{pmatrix}
					0&0\\
					-\bar c_je^{2i\theta(z)}&0
				\end{pmatrix},\\
				&\underset{z=-z_k}{{\rm Res}} M(z)=\underset{z\to -z_k}{\lim}M(z)\begin{pmatrix}
					0&c_jje^{-2i\theta(\bar z)}\\
					0&0
				\end{pmatrix},
				\end{align*}
	where $\theta(z)=xz+4tz^3$. The \( N \)-soliton solution \( q_N(x,t) \) can be  determined by
	\begin{equation}\label{eq:2.28}
		q_N(x,t) = 2i \lim_{|z| \to \infty} (zM(z))_{12}.
	\end{equation}
		\end{enumerate}
	\end{problem}

We initially employ a transformation to eliminate the singularities of \( M(z) \). Specifically, let \( \gamma \) be a closed, anticlockwise-oriented contour that encompasses all the breathers \(\{(z_k, -\bar{z}_k)\}_{k=1}^{N_1}\) and solitons \(\{\kappa_j\}_{j=1}^{N_2}\) in the upper half-plane, with \( \mathcal{D} \) denoting the finite domain bounded by \( \partial\mathcal{D}= \gamma \). Similarly, we define \( \bar{\gamma} \) as a closed, clockwise-oriented contour that encircles all the breathers \(\{(\bar{z}_k, -z_k)\}_{k=1}^{N_1}\) and solitons \(\{\bar{\kappa}_j\}_{j=1}^{N_2}\) in the lower half-plane, where \( \bar{\mathcal{D}} \) represents the finite domain delimited by \( \bar{\gamma} \).

Subsequently, we reformulate the jump conditions into the following RH problem.
\begin{problem}
		Find an analytic function $M(z):\mathbb{C}\,\setminus (\gamma\,\cup\,\bar\gamma)\to SL_2(\mathbb{C})$ with the following properties.
		
		\begin{enumerate}
			\item $ M(z)=I+\mathcal{O}(z^{-1})$ as $z\to\infty$.
		
			\item For $z\in\gamma\,\cup\,\bar\gamma$, the boundary values $M_\pm(z)$ satisfy the jump relation
			\begin{equation*}
					M_+(z)=M_-(z)V(z),
			\end{equation*}
			where
	\begin{equation*}
		V(z)=\begin{cases}
			\begin{pmatrix}
				1&0\\
				-\left(\sum_{k=1}^{N_1}\frac{c_ke^{2i\theta(z)}}{z-z_k}-\sum_{k=1}^{N_1}\frac{\bar c_ke^{2i\theta(z)}}{z+\bar z_k}+\sum_{j=1}^{N_2}\frac{h_je^{2i\theta(z)}}{z-\kappa_j}\right)&1
\end{pmatrix},\,\,z\in\gamma,\\[18pt]
			\begin{pmatrix}
1&\sum_{k=1}^{N_1}\frac{\bar c_ke^{-2i\theta(\bar z)}}{z-\bar z_k}-\sum_{k=1}^{N_1}\frac{c_ke^{2i\theta(z)}}{z+z_k}+\sum_{j=1}^{N_2}\frac{\bar h_je^{2i\theta(z)}}{z-\bar\kappa_j}\\
				0&1
\end{pmatrix},\qquad z\in\bar\gamma.
		\end{cases}
	\end{equation*}
	
		\end{enumerate}
	\end{problem}
	
We consider the limit as the number of solitons approaches infinity, with their point spectrum \( \mathcal{Z} \) uniformly filling the domain \( \mathcal{D} \). It is assumed that the norming constants \(c_k\) and \(h_j\) are scaled by a smooth function \(\alpha_\ell(z, \bar{z})\) for \(\ell = 1, 2\) as
		 \begin{equation}
c_k=\frac{S}{\pi N_1}\alpha_1(z_k,\bar z_k)\chi_{\mathcal{D}}(z_k),\quad
h_j=\frac{S}{\pi N_2}\alpha_2(\kappa_j,\bar \kappa_j)\chi_{\mathcal{D}}(\kappa_j),	
	\end{equation}
where $\chi_{\mathcal{D}}$ denotes the characteristic function of the domain $\mathcal{D}$, and let $S$ represent the area of $\mathcal{D}$. On the physical side, scaling the norming constants to be small implies that the individual solitons are centered at positions that are logarithmically large in $N$. Consequently, in the finite part of the $(x,t)$-plane, only the tails of the solitons contribute to the sum \cite{GE20}.

Upon taking the limit $N\to\infty$, we get
	\begin{align}
	&\sum_{j=1}^{N_1}\frac{c_k}{z-z_k}=\sum_{j=1}^{N_1}\frac{S}{\pi N_1}\frac{\alpha_1(z_k,\bar z_k)}{z-z_k}\chi_{\mathcal{D}}(z_k) \xrightarrow{N_1\to\infty} \iint_{\mathcal{D}}\frac{\alpha_1(w,\bar w)}{z-w}\frac{d^2w}{\pi},\\
	&\sum_{j=1}^{N_2}\frac{h_j}{\kappa-\kappa_j}=\sum_{j=1}^{N_2}\frac{S}{\pi N_2}\frac{\alpha_2(\kappa_j,\bar\kappa_j)}{\kappa-\kappa_j}\chi_{\mathcal{D}}(\kappa_j) \xrightarrow{N_2\to\infty} \iint_{\mathcal{D}}\frac{\alpha_2(w,\bar w)}{z-w}\frac{d^2w}{\pi},	
\end{align}
	where $d^2w=\frac{dw\land\bar w}{2i}=dxdy$ is the usual area element.
		
	Consequently, we arrive at a limiting RH problem as following.
	\begin{problem}
	Find an analytic function $M^\infty(z):\mathbb{C}\,\setminus (\gamma \cup \bar\gamma)\to SL_2(\mathbb{C})$ with the following properties.
	
	\begin{enumerate}
		\item $ M^\infty(z)=I+\mathcal{O}(z^{-1})$ as $z\to\infty$.
	\item For $z\in\gamma \cup \bar\gamma$, the boundary values $M^\infty_\pm(z)$ satisfy the jump relation
		\begin{equation}
			M^\infty_+(z)=M^\infty_-(z)V^\infty(z),
		\end{equation}
		where
		\begin{equation*}
		V^\infty(z)=\begin{cases}
			\begin{pmatrix}
				1&0\\
			\iint_{\mathcal{D}}\frac{e^{2i\theta(w)}(\alpha_1(w,\bar w)-\overline{\alpha_1(\bar w,w)}+\alpha_2(w,\bar w))}{w-z}\frac{d^2w}{\pi}&1
			\end{pmatrix},\,\,\,z\in\gamma,\\[15pt]
			\begin{pmatrix}
				1&\iint_{\bar{\mathcal{D}}}\frac{e^{-2i\theta(w)}(\overline{\alpha_1(\bar w,w)}-\alpha_1(w,\bar w)+\overline{\alpha_2(\bar w, w)})}{z-w}\frac{d^2w}{\pi}\\
				0&1
			\end{pmatrix}, z\in\bar\gamma.
		\end{cases}
	\end{equation*}
	\end{enumerate}
\end{problem}
The limiting mKdV soliton is derived from the matrix $M^\infty(z)$ as follows:
\begin{equation}\label{eq2.9}
q_\infty(x,t) = 2i \lim_{|z| \to \infty} (zM^\infty(z))_{12}.
\end{equation}

\subsection{Shielding of soliton gas for Quadrature domain}\label{sec222}
We initiate our analysis by considering the class of quadrature domains defined by:
\begin{equation}\label{eq:2.11}
\mathcal{D} = \left\{ z \in \mathbb{C} \mid \left| (z - d_0)^m - d_1 \right| < \rho \right\}, \quad m \in \mathbb{N},
\end{equation}
where \( d_0 \in \mathbb{C}^+ \) and \( |d_1|, \rho > 0 \) are sufficiently small to ensure that \( \mathcal{D} \subset \mathbb{C}^+ \). The parameter \( m \) is associated with the symmetry of the domain.
For a domain \( \mathcal{D} \) that is sufficiently smooth and simply connected, its boundary \( \partial\mathcal{D} \) can be described using the Schwarz function \( S(z) \) associated with \( \mathcal{D} \), given by the equation:
\[
\bar{z} = S(z), \quad S(z) = \bar{d}_0 + \left( \bar{d}_1 + \frac{\rho^2}{(z - d_0)^m - d_1} \right)^{1/m}.
\]

	For a quadrature domain \( \mathcal{D} \) and smooth functions \( \alpha_\ell(z, \bar{z}) \) for \( \ell = 1,2 \), the class of solutions to the mKdV equation remains unexplored. In the specific case where
	\begin{equation}\label{eq:ww}
		\alpha_\ell(z, \bar{z}) = N(\bar{z} - \bar{d}_{0})^{N-1}\varpi_\ell(z),
	\end{equation}
 with \( \varpi_\ell(z) \) being analytic in \( \mathcal{D} \), we can apply Green's theorem for \( z \notin \mathcal{D} \) and derive
\begin{equation}\label{eq:2.10}
\iint_{\mathcal{D}} \frac{e^{2i\theta(w)} \alpha_\ell(z, \bar{z}) \, d^2w}{\pi(z - w)} = \int_{\partial\mathcal{D}} \frac{\varpi_\ell(w)(\bar{w} - \bar{d}_0)^n e^{2i\theta(w)}}{z - w} \frac{dw}{2\pi i},
\end{equation}
and a similar expression holds for the integral over the complement of \( \mathcal{D} \), denoted by \( \bar{\mathcal{D}} \).


			\begin{theorem}	\label{th1.1}
			Let the quadrature domain \( \mathcal{D}\) be defined as in equation \eqref{eq:2.11}. Then, the gas solution \( q_\infty(x,t)\) given by \eqref{eq2.9} coincides with the \( N\)-soliton solution \( q_N(x,t)\) provided in \eqref{eq:2.28}, which is the solution set of the equation
			\[
			(z - d_{0})^N = d_{1}.
			\]
			The corresponding norming constants are given by
			\begin{align*}
				&c_k= \frac{\rho^2 \varpi_1(z_k)}{\prod_{s \neq k} (z_k - z_s)}, \quad k = 1, \dots, N_1,\\
				&\kappa_j = \frac{\rho^2 \varpi_2(\kappa_j)}{\prod_{m \neq j} (\kappa_j - \kappa_m)}, \quad j = 1, \dots, N_2.
			\end{align*}
			Here, \(\varpi_j\) for \(j=1,2\) are provided in equation \(\eqref{eq:ww}\).
		\end{theorem}

\begin{proof}
  The solution \( q_\infty(x,t) \) given in equation \eqref{eq2.9} can be expressed as
\[
M^\infty = I + \frac{1}{2\pi i} \int_{\gamma \cup \bar{\gamma}} \frac{\mu(s)(V^\infty(s) - I)}{s - z} \, ds,
\]
where
\[
\mu = (I - C_w)^{-1} I, \quad C_w f = C_-(f(V^\infty - I)),
\]
and the Cauchy projection operator is defined as
\[
(C_\pm f)(z) = \lim_{z' \to z} \frac{1}{2\pi i} \int_{\gamma \cup \bar{\gamma}} \frac{f(s)}{s - z'} \, ds.
\]

	Initially, we examine the scenario to solitons. Employing a similar approach, we can derive the corresponding solution for breathers. This solution is extracted from  \eqref{eq:2.10} by setting \(m = N_2\) in  \eqref{eq:2.11}. Subsequently, we substitute \(\bar{w} = S(w)\) into the contour integral \eqref{eq:2.10} and invoke the residue theorem at the \(N_2\) poles, which are determined by the solutions \(\{\kappa_j\}_{j=1}^{N_2}\) of the equation \((z - d_0)^{N_2} = d_1\). As a result, we obtain the following:
\begin{align}\nonumber
	& \int_{\gamma} \frac{\left(\bar{w}-\bar{d}_0\right)^{N_2} \varpi_2(w) \mathrm{e}^{2i\theta(w)}}{z-w} \frac{d w}{2 \pi i} \\\nonumber
	& =\int_{\gamma }\left(S(w)-\bar{d}_0\right)^{N_2} \varpi_2(w) \frac{\mathrm{e}^{2i\theta(w)}}{z-w} \frac{d w}{2 \pi i} \\\label{eq:2.18}
	& =\rho^2 \sum_{j=1}^{N} \frac{\varpi_2\left(\kappa_j\right)}{\prod_{s \neq j}\left(\kappa_j-\kappa_s\right)} \frac{\mathrm{e}^{2 it\theta\left(\kappa_j\right)}}{z-\kappa_j}, \quad z \notin \mathcal{D}.
\end{align}
The results then follow immediately.
\end{proof}

Here we provide an example by considering the simplest quadrature domain when \( N= m = 1 \). Such a domain coincides with the disk \(\mathbb{D}_\rho(z_0)\) of radius \(\rho > 0\) centered at \( \kappa_1\).
	
We then extract the one-soliton solution \eqref{eq:soliton} of the mKdV equation \eqref{eq:mkdc} with \( \kappa_1 =i\zeta \) and \( c_1 = \rho^2 r(\kappa_1) \). The phase shift is subsequently determined by
\[
x_0 = \frac{1}{2i\kappa_1} \ln \frac{2i\kappa_1}{|\rho^2 r(\kappa_1)|},
\]
where the radius \(\rho\) of the disk and the value of the function \( r(z) \) at \( \kappa_1 \) are the contributing factors to the phase shift of the one-soliton.

\subsection{Shielding of soliton gas for Elliptic domian}\label{sec2222}
In this section we consider the case where $\mathcal{D}$ is an ellipse. Assuming the foci $E_1$ and $E_2$ are on the imaginary axis, we have $E_1 = i\eta_1$ and $E_2 = i\eta_2$ with $\eta_2 > \eta_1 > 0$. The ellipse equation is:
\begin{equation}\label{eq:ellipse}
	\sqrt{{\rm Re}^2\,z+({\rm Im}\,z-\eta_1)^2}+\sqrt{{\rm Re}^2\,z + ({\rm Im}\,z - \eta_2)^2}=2\rho>0,
\end{equation}
which positioned in the upper half-plane by selecting $\rho$ to be sufficiently small. The boundary $\partial\mathcal{D}$ is characterized by the Schwartz function $S(z)$ of the ellipse, given by:
\begin{equation}
	\bar z=S(z)=\left(1-\frac{8\rho^2}{(\eta_2-\eta_1)^2}\right)(z+\frac{\eta_1+\eta_2}{2i})+2\frac{\rho}{c^2}\sqrt{\rho^2-c^2}\tilde S(z),
\end{equation}
where $\tilde S(z)=\sqrt{(z-i\eta_1)(z-i\eta_2)}$. The function \( S(z) \) is analytic in \( \mathbb{C} \setminus [i\eta_1, i\eta_2] \) with the line segment oriented upwards.

In such a case, we choose \( \alpha_\ell(z, \bar{z}) = \varpi_\ell(z) \) for $\ell=1,2$ to be analytic in \( \mathcal{D} \). Consequently, the integral along the boundary \( \partial\mathcal{D} \) of the ellipse in equation \(\eqref{eq:2.10}\) can be deformed to a line integral on the segment \( \mathcal{I} = [i\eta_1, i\eta_2] \), that is:

\begin{equation}
	\int_{\partial\mathcal{D}}\frac{\alpha_\ell(w)\bar we^{2i\theta(w)}}{z-w}\frac{dw}{2\pi i}=\int_{\mathcal{I}}\frac{\varpi_\ell(w)\delta S(w)e^{2i\theta(w)}}{z-w}\frac{dw}{2\pi i},\quad \ell=1,2,
\end{equation}
where $\delta S(z)=S_+(z)-S_-(z)$. Next we define
\begin{equation}\label{eq:transT}
	T(z)=\begin{cases}
		M^\infty(z),\qquad\quad z\in\mathbb{C}\setminus\{\mathcal{D}\cup\bar{\mathcal{D}}\},\\
		M^\infty(z)J(z),\quad\, z\in\mathcal{D}\cup\bar{\mathcal{D}},
	\end{cases}
\end{equation}
where

		\begin{equation}
	J(z)=\begin{cases}
		\begin{pmatrix}
			1&0\\
			\int_{\mathcal{I}}\frac{(\varpi_1(w)-\overline{\varpi_1(\bar w)}+\varpi_2(w))\delta S(w)e^{2i\theta(w)}}{z-w}\frac{dw}{2\pi i}&1
		\end{pmatrix},\quad\,\, z\in\mathcal{D},\\
		\begin{pmatrix}
			1&\int_{\bar{\mathcal{I}}}\frac{e^{-2i\theta(w)}\overline{ \varpi_1(\bar w)}-\varpi_1(w)+\overline{\varpi_2(\bar w)}\delta\overline{ S(\bar w)}}{w-z}\frac{dw}{2\pi i}\\
			0&1
		\end{pmatrix},\quad z\in\bar{\mathcal{D}},
	\end{cases}
\end{equation}

Consequently, we arrive at the following RH problem for \( T(z) \).
\begin{problem}\label{RH2.4}
	Find an analytic function $T(z):\mathbb{C}\setminus(\mathcal{I} \cup \bar{\mathcal{I}})\to SL_2(\mathbb{C})$ with the following properties.
	
	\begin{enumerate}
		\item $ T(z)=I+\mathcal{O}(z^{-1})$ as $z\to\infty$.
		\item For $z\in\mathcal{I} \cup \bar{\mathcal{I}}$,
		the boundary values $T_{\pm}(z)$ satisfy the jump relation
		$$T_+(z)=T_-(z)J(z),$$
		where
		\begin{equation}\label{eq:jumpp}
			J(z)=\begin{cases}
				\begin{pmatrix}
					1&0\\
					-r(z)e^{-2i\theta(z)}&1
				\end{pmatrix},\quad z\in\mathcal{I},\\[15pt]
				\begin{pmatrix}
					1&\overline{ r(\bar z)} e^{2i\theta(z)}\\
					0&1
				\end{pmatrix},\qquad\,\, z\in\bar{\mathcal{I}},
			\end{cases}
		\end{equation}
		where $r(z)=\delta S(z)(\varpi_1(w)-\overline{\varpi_1(\bar w)}+\varpi_2(w))$.
	\end{enumerate}
\end{problem}


\begin{theorem}	\label{th1.2}
				Let the ellipse domain \( \mathcal{D} \) be defined as in equation \(\eqref{eq:ellipse}\). Consequently, the gas solution \( q_\infty(x,t) \), as given in equation \(\eqref{eq2.9}\), corresponds to an infinite number of spectral points that are uniformly distributed along the segments \( \mathcal{I} \cup \bar{\mathcal{I}} \).
			\end{theorem}

\begin{proof}
 By the process described in equation \eqref{eq:transT}, the function \( T(z) \) transforms the jump on \( \gamma \cup \bar{\gamma} \) to \( \mathcal{I} \cup \bar{\mathcal{I}} \), with both contours oriented upwards. As a result, \( T(z) \) is analytic in \( \mathbb{C} \setminus (\mathcal{I} \cup \bar{\mathcal{I}}) \). In this context, we consider the mKdV soliton shielding, which is characterized by a jump across the contours \( \mathcal{I} \) and \( \bar{\mathcal{I}} \). This implies that the problem on the domain can be reduced to a classical RH problem with discontinuities along a collection of arcs. This reduction is also the effect of soliton shielding because the effective soliton charge can be localized from a two-dimensional domain to a collection of arcs.
\end{proof}

	\section{Step-like oscillatory initial data}\label{secinitial}
In this section, we examine the scenario where the domain $\mathcal{D}$ is an elliptic region with foci located at $i\eta_1$ and $i\eta_2$ on the imaginary axis. For $t=0$, the initial condition $q(x,0)$, which is associated with the solution of the RH problem \eqref{RH2.4}, is observed to exhibit a step-like oscillatory behavior.

\subsection{Quiescent background}
We now concentrate on the RH problem \ref{RH2.4} at \( t = 0 \) as \( x \to -\infty \). In this limit, we observe that \( e^{-2i\theta(z)} = \mathcal{O}(e^{2\eta_1 x}) \). Consequently, the jump matrix \( J(z) \) defined in \eqref{eq:jumpp} fulfills the condition
\[
\|J(z) - I\| = \mathcal{O}(e^{-c|x|}),
\]
where \( c \) is a positive constant.
where the jumps are exponentially close to the identity matrix. The small norm theory guarantees that the solution is quiescent in such a domain \cite{IZ13,TA20}, which justifies the claim that the gas is initially supported on a right half-line. As a consequence, the solution of the RH problem \ref{RH21} satisfies
\begin{equation}\label{fuwuqiong}
  M(z) = I + \mathcal{O}(e^{-c_-|x|}),\quad c_->0.
\end{equation}
This implies that the potential \( q(x,0) = \mathcal{O}(e^{-c_-|x|}) \) as \( x \to -\infty \).

\subsection{Soliton part}
We now focus on the RH problem \ref{RH2.4} at \( t = 0 \) as \( x \to +\infty \). Two growing bands \( \mathcal{I} \cup \bar{\mathcal{I}} \) emerge from the endpoints \( \pm i\eta_1 \), on which the exponential terms are asymptotically large.

\subsubsection{The $g$-function and the $f$-function}
The first step in the asymptotic analysis of the RH problem \ref{RH2.4} is to construct scalar functions \( g(z) = g(x,t,z) \) and \( f(z) = f(x,t,z) \) that control the terms with exponential growth in the jumps \eqref{eq:jumpp}.
 We also introduce a two-sheeted Riemann surface \( \mathcal{R} \) of genus one associated with the multivalued function \( R(z) \), namely:
 \begin{equation}\label{eqR}
 	\mathcal{R}=\{(w,z)\in\mathbb{C}^2\,|\,w^2=R^2(z)=(z^2+\eta_1^2)(z^2+\eta_2^2)\}.
 \end{equation}
The first sheet of the surface is identified with the sheet where \( R(z) > 0 \) for \( \text{Im}\,z = 0 \). To construct a complex Riemann surface \( \mathcal{R} \), we add the two points at infinity, \( \infty^{+\infty} \) and \( \infty^{-\infty} \), where \( \infty^{+\infty} \) is on the first sheet and \( \infty^{-\infty} \) is on the second sheet.  We fix a canonical homology basis on \( \mathcal{R} \) by choosing \( b \) to encircle \( \mathcal{I} \) clockwise on the first sheet, and \( a \) to pass from the positive side of \( \bar{\mathcal{I}} \) to \( \mathcal{I} \) on sheet 1, and from the negative side of \( \mathcal{I} \) to \( \bar{\mathcal{I}} \) on the second sheet. See Figure \ref{FF3.1}.
 \begin{figure}[htbp]
 	\begin{center}
 		\begin{tikzpicture}[scale=0.65]
 			\draw[dashed] (-5,-3)--(-5,2);
 			\draw[dashed] (-2,-3)--(-2,2);
 			\draw[dashed](0,2)--(0,-3);
 			\draw[dashed](3,2)--(3,-3);
 			\draw(-6,3)--(6,3);
 			\draw(6,3)--(4,1);
 			\draw(4,1)--(-8,1);
 			\draw(-8,1)--(-6,3);
 			\draw(-6,-2)--(6,-2);
 			\draw(6,-2)--(4,-4);
 			\draw(4,-4)--(-8,-4);
 			\draw(-8,-4)--(-6,-2);
 			\draw (-5,2)--(-2,2);
 			\draw(0,2)--(3,2);
 			\draw(-5,-3)--(-2,-3);
 			\draw(0,-3)--(3,-3);
 			\draw[blue](1.7,2) ellipse(2.5 and 0.5);
 			\draw[blue][->](-0.8,2) arc(180:90:2.5 and 0.5);
 			\draw[blue][->](4.2,2) arc(360:270:2.5 and 0.5);
 			\draw[red](0.8,2)  arc(0:180:2 and 0.5) ;
 			\draw[red][->](-3.2,2)  arc(180:90:2 and 0.5) ;
 			\draw[dashed][red](0.8,-3)  arc(360:180:2 and 0.5) ;
 			\draw[dashed][red][->](0.8,-3)  arc(360:270:2 and 0.5) ;
 			
 			\node  [below]  at (-5.6,2.2 ) {\scalebox{0.7}{$-i\eta_2$}};
 			\node  [below]  at (-1.8,2.2 ) {\scalebox{0.7}{
 					$-i\eta_1$}};
 			\node  [below]  at (-0.25,2.3) {\scalebox{0.7}{$i\eta_1$}};
 			\node  [below]  at (3.4,2.3) {\scalebox{0.7}{$i\eta_2$}};
 			\node  [below]  at (-5.5,-2.8 ) {\scalebox{0.7}{$-i\eta_2$}};
 			\node  [below]  at (-1.8,-2.8 ) {\scalebox{0.7}{$-i\eta_1$}};
 			\node  [below]  at (0,-2.8) {\scalebox{0.7}{$i\eta_1$}};
 			\node  [below]  at (3.4,-2.8) {\scalebox{0.8}{$i\eta_2$}};
 			\node [thick] [below]  at (-1,3) {\scalebox{0.8}{$a$}};
 			\node [thick] [below]  at (3, 3) {\scalebox{0.8}{$b$}};
 			\node [thick] [below]  at (5,3) {\scalebox{0.8}{$\infty^+$}};
 			\node [thick] [below]  at (5,-2) {\scalebox{0.8}{$\infty^-$}};
 		\end{tikzpicture}
 \caption{Construction of the genus-one Riemann surface $\mathcal{R}$ and its basis of cycles.}

 \end{center}
 \label{FF3.1}
 \end{figure}
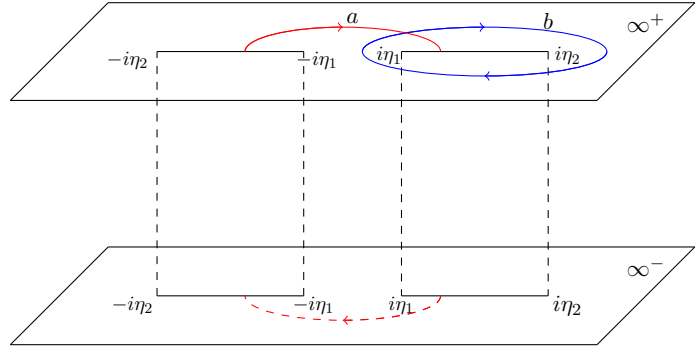

 On the genus-one Riemann surface $\mathcal{R}$ defined by \eqref{eqR} let
 \begin{equation}
 	\omega=\left(\oint_a\frac{dz}{R(z)}\right)^{-1}\frac{dz}{R(z)}=\frac{\eta_2}{4iK(m)}\frac{dz}{R(z)},\quad  m=\eta_1^2/\eta_2^2,
 \end{equation}
 where $K(m)=\int_0^1\frac{ds}{\sqrt{1-m\sin^2\theta}}$ is the complete elliptic integral of the first kind. We define the modulus
 \begin{equation}
 	\tau=\oint_b\omega=\frac{iK(1-m)}{2K(m)}.
 \end{equation}

 We recall the definition of the Jacobi elliptic function
 \begin{equation}
 	\theta(z,\tau)=\sum_{n\in\mathbb{Z}}e^{i\pi n^2\tau+2i\pi nz},\quad z\in\mathbb{C},
 \end{equation}
 which is an even function of $z$ and satisfies the periodicity conditions
 \begin{equation}
 	\theta(z+h+k\tau;\tau)=e^{-i\pi k^2\tau-2i\pi kz}\theta(z;\tau),\quad k,h\in\mathbb{Z},
 \end{equation}
 and has a simple zero at the half period $\frac{\tau}{2}+\frac{1}{2}$.

 Now we introduce the following new matrix function

\begin{equation}
	T^{(1)}(z)=e^{ig_0(x,t)\sigma_3}
    T(z)e^{-ig(z)\sigma_3}f(z)^{\sigma_3},
\end{equation}
where $g_0$ satifies \eqref{g0} and $\sigma_3=\begin{pmatrix}
	1&0\\
	0&-1\end{pmatrix}$, $g(z)$ and $f(z)$ defined by
\begin{align}\label{eq:3.6}
	&g(z)=\frac{R(z)}{2\pi i}\left(\int_{\mathcal{I}\,\cup\,\bar{\mathcal{I}}}\frac{-2xs-8ts^3}{R_+(s)}\frac{ds}{s-z}-\int_{-i\eta_1}^{i\eta_1}\frac{\Omega}{R(s)}\frac{ds}{s-z}\right)	
	,\\\label{eq:3.7}
	&f(z)=\exp\left\{\frac{R(z)}{2\pi i}\left(\int_{\mathcal{I}}\frac{-\log r(s)}{R_+(s)(s-z)}ds+\int_{\bar{\mathcal{I}}}\frac{\log \overline{r(\bar s)}}{R_+(s)(s-z)}d\zeta+\int_{-i\eta_1}^{i\eta_1}\frac{i\Delta}{R(s)(s-z)}ds\right)\right\},
\end{align}
where
\begin{align}
\Omega=\Omega(x,t)=\frac{\pi\eta_2}{K(m)}(x-2(\eta_1^2+\eta_2^2)t),\quad m=\eta_1^2/\eta_2^2,
\end{align}
and
\begin{equation}\label{delta}
	\Delta=-i\left(\int_0^{i\eta_1}\frac{ds}{R(s)}\right)^{-1}\left(\int_{\mathcal{I}}\frac{\ln r(s)}{R_+(s)}ds\right).
\end{equation}
\begin{proposition}\label{proposition}
The functions \( f(z) \) and \( g(z) \), defined in equations \eqref{eq:3.6} and \eqref{eq:3.7}, possess the following properties.
	\begin{enumerate}
		\item $g(z)$ is analytic in $\mathbb{C}\setminus[-i\eta_2,i\eta_2]$ and
		\begin{align*}
			&g_+(z)+g_-(z)=-8z^3t-2zx,\quad z\in\mathcal{I}\cup\bar{\mathcal{I}},\\
			&g_+(z)-g_-(z)=-\Omega(x,t),\qquad\,\,\,\, z\in(-i\eta_1,i\eta_1).
		\end{align*}
		
			\item $f(z)$ is analytic in $\mathbb{C}\setminus[-i\eta_2,i\eta_2]$ and
		\begin{align*}
			&f_+(z)f_-(z)=\frac{1}{r(z)},\quad z\in\mathcal{I},\\
				&f_+(z)f_-(z)=\overline{ r(\bar z)},\quad\, z\in\bar{\mathcal{I}},\\
			&\frac{f_+(z)}{f_-(z)}=e^{i\Delta},\qquad\quad\,\,\,\, z\in[-i\eta_1,i\eta_1].
		\end{align*}
		\item Schwarz symmetry
		\begin{equation}
			\overline{g(\bar z)}=g(z),\quad \overline{f(\bar z)}=f^{-1}(z),\quad f(-z)=f^{-1}(z)
		\end{equation}
		\item As $z\to\infty$, we have
		\begin{align}\label{g0}
			&g(z)=g_0(x,t)+\mathcal{O}(z^{-1}),\\
			&f(z)=1+\mathcal{O}(z^{-1})
		\end{align}
		\item Near the endpoints $z=\pm\eta_j,j=1,2$, we have
		\begin{equation}
	g(z)+4z^3t+zx=\mathcal{O}(z\pm i\eta_j)^{1/2},\quad z\to\pm i\eta_j.
		\end{equation}
	\end{enumerate}
\end{proposition}
Letting
\begin{equation}
	\varphi(z)=g(z)+zx+4z^3t,
\end{equation}
from Proposition \ref{proposition}, we have the following relations:
\begin{itemize}
	\item $\varphi_+(z)+\varphi_-(z)=0$ for $z\in\mathcal{I}\,\cup\,\bar{\mathcal{I}}$,
	\item   $\varphi_+(z)-\varphi_-(z)=-\Omega<0$ for $z\in [-i\eta_1,i\eta_1]$,
	\item $\varphi(z)=4tz^3+xz+\mathcal{O}(z^{-1})$ as $z\to\infty$.
\end{itemize}

Then we obtain the following Riemann-Hilbert problem for \( T^{(1)}(z) \):
\begin{problem}
	Find an analytic function $T^{(1)}(z):\mathbb{C}\setminus[-i\eta_2,i\eta_2]\to SL_2(\mathbb{C})$ with the following properties.
	
	\begin{enumerate}
		\item $ T^{(1)}(z)=I+\mathcal{O}(z^{-1})$ as $z\to\infty$.
		\item For $z\in [-i\eta_2,i\eta_2]$, the boundary values $T^{(1)}_{\pm}(z)$ satisfy the jump relation
		$$T^{(1)}_+(z)=T^{(1)}_-(z)V^{(1)}(z),$$
    where
\begin{equation}\label{eqV2}
	V^{(1)}=\begin{cases}
		\begin{pmatrix}
			\frac{e^{2i\varphi_-(z)}}{r(z)f^2_-(z)}&0\\
			1& 	\frac{e^{2i\varphi_+(z)}}{r(z)f^2_+(z)}
			 \end{pmatrix},\quad z\in\mathcal{I}\\[15pt]
		\begin{pmatrix}
			\frac{e^{-2i\varphi_+(z)}}{\overline{r(\bar z)}f^2_+(z)}&1\\
			0& 	\frac{e^{-2i\varphi_-(z)}}{\overline{r(\bar z)}f^2_-(z)}   \end{pmatrix},\quad z\in\bar{\mathcal{I}}\\[15pt]
		\begin{pmatrix}
			e^{i(\Omega+\Delta)}&0\\
			0&e^{-i(\Omega+\Delta)}    \end{pmatrix},\quad\,\,\, z\in[-i\eta_1,i\eta_1].
	\end{cases}
\end{equation}
	\end{enumerate}
\end{problem}
\subsubsection{Opening lenses in the support of the soliton gas}
We introduce one further transformation which ``opening  lenses'' away from the $\mathcal{I} \cup \bar{\mathcal{I}}$, having the effect of deforming the oscillatory jumps \eqref{eqV2} onto new contours where they are exponentially decaying.

We start by defining the analytic continuation $\hat{r}(z)$ of the function $r(z)$ off the interval $\mathcal{I} \cup \bar{\mathcal{I}}$ with the requirement that
$$
\hat{r}_\pm(z) = \pm r(z), \quad z \in \mathcal{I} \cup \bar{\mathcal{I}},
$$
 Consequently, we introduce the open lenses as depicted in Figure \ref{lense}.
 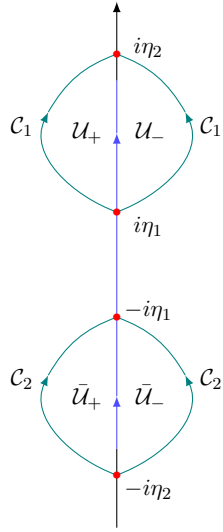
\begin{figure}[htp]
	\begin{center}
		\begin{tikzpicture}[scale=0.7]
			\draw(0,-5.5)--(0,-4);
				\draw[blue!70][-latex](0,-4)--(0,-3);
			\draw[blue!70][-latex](0,-3)--(0,2);
			\draw[blue!70](0,2)--(0,3);
			\draw[-latex](0,3)--(0,4.5);
			\draw  [teal][-latex](0,-4.5) to [out=20,in=-65](1.3,-2.6);
			\draw [teal](1.3,-2.6) to [out=120,in=-20](0,-1.5);
			\draw[teal] [-latex](0,-4.5) to [out=160,in=-115](-1.3,-2.6);
			\draw  [teal](-1.3,-2.6) to [out=60,in=-160](0,-1.5);
			\draw  [teal][-latex](0,0.5) to [out=20,in=-65](1.3,2.4);
			\draw [teal] (1.3,2.4) to [out=120,in=-20](0,3.5);
			\draw  [teal][-latex](0,0.5) to [out=160,in=-115](-1.3,2.4);
			\draw  [teal](-1.3,2.4) to [out=60,in=-160](0,3.5);
			\draw (0.6,-4.4)node[below]{\scalebox{0.8}{$-i\eta_2$}};
			\draw(0.6,-1)node[below]{\scalebox{0.8}{$-i\eta_1$}};
			\draw(0.6,4)node[below]{\scalebox{0.8}{$i\eta_2$}};
			\draw(0.6,0.6)node[below]{\scalebox{0.8}{$i\eta_1$}};
			\draw(1.8,2.5)node[below]{\scalebox{0.8}{$\mathcal{C}_1$}};
			\draw(1.8,-2.3)node[below]{\scalebox{0.8}{$\mathcal{C}_2$}};
			\draw(-1.8,2.5)node[below]{\scalebox{0.8}{$\mathcal{C}_1$}};
			\draw(-1.8,-2.3)node[below]{\scalebox{0.8}{$\mathcal{C}_2$}};
			\draw(-0.55,2.4)node[below]{\scalebox{0.8}{$\mathcal{U}_+$}};
			\draw(0.65,2.4)node[below]{\scalebox{0.8}{$\mathcal{U}_-$}};
			\draw(-0.55,-2.6)node[below]{\scalebox{0.8}{$\bar{\mathcal{U}}_+$}};
			\draw(0.65,-2.6)node[below]{\scalebox{0.8}{$\bar{\mathcal{U}}_-$}};
				\node [red]at(0,-4.5) {$\scriptscriptstyle\bullet$};
			\node [red]at(0,-1.5) {$\scriptscriptstyle\bullet$};
			\node [red]at(0,0.5) {$\scriptscriptstyle\bullet$};
			\node [red]at(0,3.5) {$\scriptscriptstyle\bullet$};
		\end{tikzpicture}
	\end{center}
	\caption{As $x \to +\infty$, the lenses $\mathcal{U}_\pm \cup \bar{\mathcal{U}}_\pm$ and the jump matrix $V^{(2)}(z)$ tend to the identity matrix on the lens contours $\mathcal{C}_j$, for $j=1,2$, which are depicted by the green lines. The jump matrices on $\mathcal{I} \cup \bar{\mathcal{I}}$ coincide with those of the outer model problem.}
	\label{lense}
\end{figure}	

We factorize the jump matrices \eqref{eqV2}
for $z \in \mathcal{I}$, we have
$$
V^{(1)}(z)=\begin{pmatrix}
1 & \frac{e^{2i\varphi_-(x,t,z)}}{r(z)f^2_-(z)}\\
0 & 1
\end{pmatrix}\begin{pmatrix}
0 &- 1 \\
1 & 0
\end{pmatrix}\begin{pmatrix}
1 &  \frac{e^{2i\varphi_+(x,t,z)}}{r(z)f^2_+(z)}\\
0& 1
\end{pmatrix},
$$
for $z \in \bar{\mathcal{I}}$, we obtain
$$
V^{(1)}(z)=\begin{pmatrix}
1 &0\\
 \frac{e^{-2i\varphi_-(x,t,z)}}{\overline{r(\bar z)}f^2_-(z)} & 1
\end{pmatrix}\begin{pmatrix}
0 & 1 \\
-1 & 0
\end{pmatrix}\begin{pmatrix}
1 &0 \\
 \frac{e^{-2i\varphi_+(x,t,z)}}{\overline{r(\bar z)}f^2_+(z)}& 1
\end{pmatrix}.
$$
We introduce a further transformation of the problem:

\begin{equation}\label{G(z)trans}
	T^{(2)}(z, x)=T^{(1)}(z, x) G(z, x)
\end{equation}

where
$$
G(z,x)= \begin{cases}\begin{pmatrix}
		1 & \frac{e^{2i\varphi(z)}}{\hat r(z)f^2(z)}\\
		0 & 1
	\end{pmatrix}, z\in  \mathcal{U}_+,
\begin{pmatrix}
	1 &  -\frac{e^{2i\varphi(z)}}{\hat r(z)f^2(z)}\\
	0& 1
\end{pmatrix}, z\in\mathcal{U}_-
\\[15pt]
\begin{pmatrix}
	1 &0\\
	\frac{e^{-2i\varphi(z)}}{\overline{\hat r(\bar z)}f^2(z)} & 1
\end{pmatrix},  z\in  \bar{\mathcal{U}}_+
\begin{pmatrix}
	1 &0\\
-	\frac{e^{-2i\varphi(z)}}{\overline{\hat r(\bar z)}f^2(z)} & 1
\end{pmatrix}\,  z\in  \bar{\mathcal{U}}_- \\
I,  \text { elsewhere. }\end{cases}
$$

It follows easily that $T^{(2)}(z)$ satisfies the following RH problem.
\begin{problem}
	Find an analytic function $T^{(2)}(z):\mathbb{C}\setminus([-i\eta_2,i\eta_2] \cup \mathcal{C}_1 \cup \mathcal{C}_2)\to SL_2(\mathbb{C})$ with the following properties.
	
	\begin{enumerate}
		\item $ T^{(2)}(z)=I+\mathcal{O}(z^{-1})$ as $z\to\infty$.
		\item For $z\in [-i\eta_2,i\eta_2] \cup \mathcal{C}_1 \cup \mathcal{C}_2$, the boundary values $T^{(2)}_{\pm}(z)$ satisfy the jump relation
		$$T^{(2)}_+(z)=T^{(2)}_-(z)V^{(2)}(z),$$
where
\begin{equation}\label{VV2}
V^{(2)}(z ; x)= \begin{cases}\begin{pmatrix}
0 & -1 \\
1& 0
\end{pmatrix}, z \in \mathcal{I} ,\,\,\,\,
\begin{pmatrix}
	1 &0\\
	\frac{e^{-2i\varphi(x,t,z)}}{\overline{\hat r(\bar z)}f^2(z)} & 1
\end{pmatrix}, z \in \mathcal{C}_2;\\[15pt]
\begin{pmatrix}
	0 & 1 \\
	-1 & 0
\end{pmatrix} ,z \in \bar{\mathcal{I}}, \,\,\,\,\begin{pmatrix}
1 & \frac{e^{2i\varphi(z)}}{\hat r(z)f^2(z)}\\
0& 1
\end{pmatrix},\hspace{0.3cm} z \in \mathcal{C}_1,
\\[15pt]
\begin{pmatrix}
e^{i(\Omega+\Delta)} & 0 \\
0 & e^{-i(\Omega+\Delta)}
\end{pmatrix},\hspace{0.2cm}z \in\left[-i \eta_1 ; i \eta_1\right].
\end{cases}
\end{equation}
	\end{enumerate}
\end{problem}
This ensures that the $z$-dependent oscillatory component is confined solely to the outer boundary of the lenses. We now proceed to investigate the sign of \(\operatorname{Im}\,\varphi(z)\) in the vicinity of the lenses.
\begin{lemma}\label{lemma31}
	The function $\varphi(z)$ satisfy the following inequalities:
\begin{align*}
		& \operatorname{Im}\,\varphi(z)>ct,\quad z\in L_1\setminus\{i\eta_1,\eta_2\},\\
		& \operatorname{Im}\,\varphi(z)<-ct ,\quad z\in L_2\setminus\{-i\eta_1,-\eta_2\}.
	\end{align*}

\end{lemma}
\begin{proof}
Using the representation of \( g(z) \) given by equation \eqref{eq:3.6}, for a fixed \( s_0 = \max\{|z|, \eta_2\} \), we apply the residue theorem to express \(\varphi(z)\) as follows:
\[
\varphi(z) = \frac{R(z)}{2\pi i} \left( \oint_{|s|=s_0} \frac{xs + 4ts^3}{R(s)} \frac{ds}{s-z} - \int_{-i\eta_1}^{i\eta_1} \frac{\Omega}{R(s)} \frac{ds}{s-z} \right)
\]
where the contour integral is positively oriented. From this expression, we derive the derivative of \(\varphi(z)\):
\[
\varphi'(z) = \frac{12t(z^2 + \eta_2^2)(z^2 + v^2)}{R(z)}, \quad v \in (0, \eta_1)
\]
which leads to the following inequalities:
\begin{align*}
	& \varphi'(z) < 0, \quad z \in (i\eta_2, i\infty) \\
	& i\varphi'(z) > 0, \quad z \in \mathcal{I}
\end{align*}
The results then follow immediately.
\end{proof}

\subsubsection{The outer model problem}
Lemma \ref{lemma31} guarantees that the off-diagonal entries of the jump matrices \eqref{VV2} along the upper and lower lenses are exponentially small in the region as \( x \to +\infty \). Consequently, these jump matrices are asymptotically close to the identity outside small neighborhoods of \( \pm i\eta_1 \) and \( \pm i\eta_2 \). Now, we arrive at the model RH problem.

\begin{problem}\label{RH3.3}
	Find an analytic function $S(z):\mathbb{C}\setminus [-i\eta_2,i\eta_2] \to SL_2(\mathbb{C})$ with the following properties.
	
	\begin{enumerate}
		\item $ S(z)=I+\mathcal{O}(z^{-1})$ as $z\to\infty$.
		\item For $z\in  [-i\eta_2,i\eta_2] $, the boundary values $S_{\pm}(z)$ satisfy the jump relation
		$$S_+(z)=S_-(z)V_S,$$
		
		where
		$$
	V_S(z ; x)= \begin{cases}\begin{pmatrix}
				0 & -1 \\
				1 & 0
		\end{pmatrix},\, z \in \mathcal{I}, \,\,\,\,
		\begin{pmatrix}
			0 & 1 \\
			-1 & 0
		\end{pmatrix},\, \, z \in  \bar{\mathcal{I}} ,\\[15pt]
		\begin{pmatrix}
				e^{i (\Omega+\Delta)} & 0 \\
				0 & e^{-i(\Omega+\Delta)}
		\end{pmatrix},\,  z \in\left[-i \eta_1 ; i \eta_1\right].\end{cases}
		$$
	\end{enumerate}
\end{problem}

	Then, we define the integral
	\begin{equation}
		A(z)=\int_{i\eta_2}^z\omega,\quad z\in\mathbb{C}\setminus[-i\eta_2,i\eta_2],
	\end{equation}
	and we observe that
	\begin{align*}
		&A(+\infty)=\frac{1}{4},\hspace{1.5cm} A_+(i\eta_1)=-\frac{\tau}{2},\\ &
		A_+(-i\eta_2)=-\frac{1}{2},\hspace{0.8cm}
		A_+(-i\eta_1)=-\frac{1}{2}-\frac{\tau}{2},
	\end{align*}
	and
	\begin{align*}
		&A_+(z)+A_-(z)=0,\qquad z\in\mathcal{I},\\
		&A_+(z)+A_-(z)=-1,\quad\,\,z\in\bar{\mathcal{I}},\\
		&A_+(z)-A_-(z)=-\tau,\quad\,\,z\in[-i\eta_2,i\eta_2].
	\end{align*}
	Next we introduce the quantity
	\begin{equation}
		\gamma(z) = \sqrt[4]{\frac{z^2 + \eta_1^2}{z^2 + \eta_2^2}}, \quad  z \in \mathbb{C} \setminus (\mathcal{I} \cup \bar{\mathcal{I}}),
	\end{equation}
and normalized such that $\gamma(z)\to1$ as $z\to\infty$. Then
\begin{align}
	&\gamma_+(z)=-i\gamma_-(z),\quad\, z\in\mathcal{I},\\
	&\gamma_+(z)=i\gamma_-(z),\qquad z\in\bar{\mathcal{I}}.
\end{align}

We can show that the solution of the model RH problem \eqref{RH3.3} has the following form
\begin{align}
	S(z)=\frac{1}{2}\frac{\theta(0;\tau)}{\theta(\frac{\Omega+\Delta}{2\pi};\tau)}\begin{pmatrix}
		\frac{\theta(A(z)+\frac{1}{4}+\varpi;\tau)}{\theta(A(z)+\frac{1}{4};\tau)}\Gamma^+(z)&	-i\frac{\theta(-A(z)+\frac{1}{4}+\varpi;\tau)}{\theta(-A(z)+\frac{1}{4};\tau)}\Gamma^-(z)\\
		i	\frac{\theta(A(z)-\frac{1}{4}+\varpi;\tau)}{\theta(A(z)-\frac{1}{4};\tau)}\Gamma^-(z)&	\frac{\theta(-A(z)-\frac{1}{4}+\varpi;\tau)}{\theta(-A(z)-\frac{1}{4};\tau)}\Gamma^+(z)
	\end{pmatrix}
\end{align}
where $\varpi=\frac{\Omega+\Delta}{2\pi}, \Gamma^+(z)=\gamma(z)+\frac{1}{\gamma(z)}$ and $\Gamma^-(z)=\gamma(z)-\frac{1}{\gamma(z)}$.

\subsubsection{The local models at the endpoints}
In this section, we analyze the error matrices associated with the endpoints $\pm i\eta_1$ and $\pm i\eta_2$ of the intervals $\mathcal{I}$ and $\bar{\mathcal{I}}$, employing a nonlinear steepest descent analysis.

Initially, we make certain assumptions regarding the behavior of \( r(z) \) in the vicinity of the points \( \pm i\eta_1, \pm i\eta_2 \). Indeed, in \cite{MT23}, the mKdV equation was studied in conjunction with a RH problem \ref{RH2.4}. It was demonstrated that if the function $r(z)$ exhibits local behavior near the endpoints $\pm i\eta_1,\pm i\eta_2$ of the form $r(z) \sim |z \pm i\eta_j|^{\pm1/2} \tilde r(z)$ for $j=1,2$, where $\tilde r(z)$ is an analytic function that is locally bounded and non-zero in a neighborhood of the endpoints, then it is feasible to adjust the lenses of the opening factorization such that the error matrices converge exponentially to the identity matrix uniformly for $z \in \mathbb{C}$.

Specifically, following our initial assumptions, we proceed to provide the analytic continuation of \( r(z) \) away from \( \mathcal{I} \).

\begin{assumption}
When \( r(z) |z - i\eta_j|^{\pm1/2} \) for \( j = 1, 2 \) is bounded and non-zero on \( \mathcal{I} \), we assume the existence of an analytic continuation \( \hat{r}(z) \) of \( r(z) \) off \( \mathcal{I} \) that satisfies the following conditions:
	\begin{enumerate}
		\item The function \( \hat{r}(z) \) is analytic in \( \mathcal{U}^h_\pm \) and satisfies \( \hat{r}(z)|_{z \in [i\eta_1, i\eta_2]} = r(z) \), where \( \mathcal{U}^h_\pm \) is depicted in Figure \ref{fixlense}.
		\item For \( z \in [i\eta_2, i(\eta_2 + h)] \cup [i(\eta_1 - h), i\eta_1] \), we have \( \hat{r}_+(z) + \hat{r}_-(z) = 0 \).
	\end{enumerate}
\end{assumption}
In these specific cases, we observe that \( r(z) = \delta S(z) \eta^2(z) \), where \( \eta^2(z) \) is bounded within the original domain \( \mathcal{D} \), and \( \delta S(z) \) is defined in \( \mathcal{I} \) with behavior at the endpoint of the form \( \delta S(z) \sim |z - i\eta_j|^{1/2} \). Under this assumption, upon applying the transformation given by equation \eqref{G(z)trans}, we identify jumps within the intervals \( [i\eta_2, i(\eta_2 + h)] \) and \( [i(\eta_1 - h), i\eta_1] \), as depicted in Figure \ref{fixlense}.
 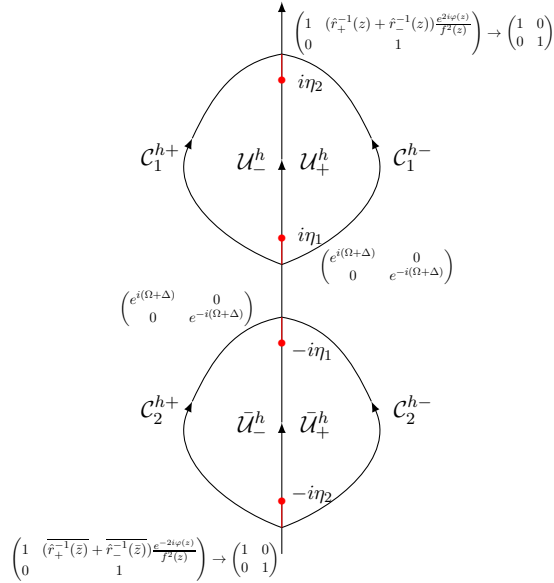
\begin{figure}[htp]
	\begin{center}
		\begin{tikzpicture}[scale=0.7]
			\draw(0,-5.5)--(0,-4.5);
			\draw(0,-1.5)--(0,0.5);
			\draw[-latex](0,3.5)--(0,5);
			\draw[-latex](0,-4.5)--(0,-3);
		    \draw(0,-3)--(0,-1.5);
		    \draw[-latex](0,0.5)--(0,2);
		    \draw(0,2)--(0,3.5);
		
			\draw [-latex](0,-5) to [out=20,in=-65](1.7,-2.6);
			\draw (1.7,-2.6) to [out=115,in=-10](0,-1);
			\draw [-latex](0,-5) to [out=160,in=-115](-1.7,-2.6);
			\draw (-1.7,-2.6) to [out=65,in=-170](0,-1);
			
			\draw [-latex](0,0) to [out=20,in=-65](1.7,2.4);
			\draw (1.7,2.4) to [out=115,in=-10](0,4);
			\draw [-latex](0,0) to [out=160,in=-115](-1.7,2.4);
			\draw (-1.7,2.4) to [out=65,in=-170](0,4);
			
			\draw[red](0,0)--(0,0.5);
			\draw[red](0,3.5)--(0,4);
			
				\draw[red](0,-5)--(0,-4.5);
			\draw[red](0,-1)--(0,-1.5);
			
			\draw(0.58,-4)node[below]{\scalebox{0.7}{$-i\eta_2$}};
			\draw(0.58,-1.3)node[below]{\scalebox{0.7}{$-i\eta_1$}};
			\draw(0.55,3.8)node[below]{\scalebox{0.7}{$i\eta_2$}};
			\draw(0.55,0.9)node[below]{\scalebox{0.7}{$i\eta_1$}};
			\draw(2.5,2.5)node[below]{\scalebox{0.8}{$\mathcal{C}^{h-}_1$}};
			\draw(2.5,-2.3)node[below]{\scalebox{0.8}{$\mathcal{C}^{h-}_2$}};
			\draw(-2.3,2.5)node[below]{\scalebox{0.8}{$\mathcal{C}^{h+}_1$}};
			\draw(-2.3,-2.3)node[below]{\scalebox{0.8}{$\mathcal{C}^{h+}_2$}};
			\draw(-0.55,2.4)node[below]{\scalebox{0.8}{$\mathcal{U}^h_-$}};
			\draw(0.65,2.4)node[below]{\scalebox{0.8}{$\mathcal{U}^h_+$}};
			\draw(-0.55,-2.6)node[below]{\scalebox{0.8}{$\bar{\mathcal{U}}^h_-$}};
			\draw(0.65,-2.6)node[below]{\scalebox{0.8}{$\bar{\mathcal{U}}^h_+$}};
			\node [red]at(0,-4.5) {$\scriptscriptstyle\bullet$};
			\node [red]at(0,-1.5) {$\scriptscriptstyle\bullet$};
			\node [red]at(0,0.5) {$\scriptscriptstyle\bullet$};
			\node [red]at(0,3.5) {$\scriptscriptstyle\bullet$};
				\draw(2,0.5)node[below]{\scalebox{0.5}{$\begin{pmatrix}
							e^{i(\Omega+\Delta)}&0\\
							0&e^{-i(\Omega+\Delta)}
						\end{pmatrix}$}};
				\draw(2.7,5)node[below]{\scalebox{0.5}{$\begin{pmatrix}
						1&(\hat r^{-1}_+(z)+\hat r^{-1}_-(z))\frac{e^{2i\varphi(z)}}{f^2(z)}\\
						0&1
					\end{pmatrix}\to\begin{pmatrix}
					1&0\\
					0&1
					\end{pmatrix}$}};
		\draw(-1.8,-0.3)node[below]{\scalebox{0.5}{$\begin{pmatrix}
							e^{i(\Omega+\Delta)}&0\\
							0&e^{-i(\Omega+\Delta)}
						\end{pmatrix}$}};
				\draw(-2.6,-5)node[below]{\scalebox{0.5}{$\begin{pmatrix}
							1&(\overline{\hat r^{-1}_+(\bar z)}+\overline{\hat r^{-1}_-(\bar z)})\frac{e^{-2i\varphi(z)}}{f^2(z)}\\
							0&1
						\end{pmatrix}\to\begin{pmatrix}
							1&0\\
							0&1
						\end{pmatrix}$}};
		\end{tikzpicture}
	\end{center}
	\caption{The new lenses $\mathcal{U}^h_\pm  \cup \bar{\mathcal{U}}^h_\pm$  and the analytic continuation jump matrix across the continuous domain (red lines) are considered. }
	\label{fixlense}
\end{figure}

We now define the Error matrix
\begin{equation}
	\mathcal{E}(z)=T^{(2)}(z)S^{-1}(z),
\end{equation}
which is analytic for $z\in\mathbb{C}\setminus(\partial(\mathcal{C}_1^{h-}\cup\mathcal{C}_1^{h+}\cup\mathcal{C}_2^{h-}\cup\mathcal{C}_2^{h+}))$ and it has the jump condition
\begin{equation}
	\mathcal{E}_+(z)=\mathcal{E}_-(z)V_{\mathcal{E}}(z),
\end{equation}
where
\begin{equation}
	V_{\mathcal{E}}(z)=V_s(z)V^{(2)}(z)V_s^{-1}(z)
\end{equation}
Since $S(z)$ is bounded in $x$, the jump matrices tend to the identity exponentially fast with respect to the matrix norm. From the small norm lemma, we have that the matrix $\mathcal{E}(z)$, uniformly in $z\in\mathbb{C}$, tends exponentially to the identity as $x\to+\infty$,
\begin{equation}
	\mathcal{E}(z)=I+\mathcal{O}(e^{-c_+x}),\quad c_+>0.
\end{equation}
 It follows that the model problem $S(z)$ coincides with $T^{(2)}(z)$ up to an exponentially small error as $x\to+\infty$.

 Now we provide the initial datum that is step-like oscillatory in the elliptical domain.

			\begin{theorem}\label{th33}
Let the ellipse domain \( \mathcal{D} \) be defined as in equation \(\eqref{eq:ellipse}\).
Then the solution of the RH problem  \ref{RH2.4}  generates an initial datum $q(x,0)$ of the mkdV equation \eqref{eq:mkdc} that is step-like oscillatory with the following behaviours at $x\to\pm\infty:$
				\begin{enumerate}
					\item As \( x \to -\infty \), \( q(x,0) \) decays as
					\[ q(x,0) = \mathcal{O}(e^{-c_-|x|}). \]
					with a positive constant \( c_- \).
					
					\item As \( x \to +\infty \), \( q(x,0) \) is given by
					\[ q(x,0) = -(\eta_2 + \eta_1) {\rm dn}\left((\eta_2 + \eta_1)(x - 2(\eta_1^2 + \eta_2^2)t - x_0); m_1\right) + \mathcal{O}(e^{c_+x}). \]
					with a positive constant \( c_+ \). The Jacobi elliptic function \( \text{dn}(z, m_1) \) has a modulus of \( m_1 \), where \( m_1 \) and \( x_0 \) are specified in equation \eqref{parameter}.
				\end{enumerate}
			\end{theorem}
\begin{proof}			

We   focus on the asymptotic behaviour for $x\to+\infty$. Tracing back the chain of transformations from the remainder problem $\mathcal{E}(z)$ to the original RH problem \ref{RH2.4} for $T(z)$, we find that
\begin{align*}
		q(x,t) &= 2ie^{ig_0(x,t)} \lim_{|z| \to \infty} (zT_{12}^{(2)}(z)e^{-ig(z)}f(z)),\\
			&= 2ie^{ig_0(x,t)} \lim_{|z| \to \infty} (zS_{12}(z))e^{-ig(z)}f(z))+\mathcal{O}(e^{c_+x}),\\
				&= -e^{ig_0(x,t)}(\eta_2-\eta_1)\frac{\theta(\frac{1}{2}+\varpi;\tau)}{\theta(\frac{1}{2};\tau)}\cdot\frac{\theta(0;\tau)}{\theta(\varpi;\tau)}
				e^{-ig(z)}f(z)+\mathcal{O}(e^{c_+x}),\\
		 &= -(\eta_2+\eta_1){\rm dn}\left((\eta_2+\eta_1)(x-2(\eta_1^2+\eta_2^2)t-x_0);m_1\right)+\mathcal{O}(e^{c_+x}),
\end{align*}
where $c_+>0$ and
\begin{equation}\label{parameter}
	m_1=\frac{4\eta_2\eta_1}{(\eta_1+\eta_2)^2},\quad x_0=\frac{K(m)(\Delta-\pi)}{\eta_2\pi}.
\end{equation}
Subsequently, by incorporating equation \eqref{fuwuqiong}, we complete the proof of Theorem \ref{th33}.
\end{proof}		

\vspace{4mm}

	\noindent\textbf{Acknowledgements}
	
	This work is supported by  the National Natural Science
	Foundation of China (Grant No. 12271104,  51879045).\vspace{2mm}
	
	\noindent\textbf{Data Availability Statements}
	
	The data that supports the findings of this study are available within the article.\vspace{2mm}
	
	\noindent{\bf Conflict of Interest}
	
	The authors have no conflicts to disclose.

	\end{document}